\documentclass[12pt]{iopart}

\usepackage{iopams,mathrsfs,cite,color}
\usepackage{todonotes,xcolor,amsthm,ulem}
\newcommand{\rem}[1]{}
\newcommand{\Rey}{\mathit{Re}}
\newcommand{\Gr}{\mathit{Gr}}
\newcommand{\bel}{\begin{equation}}
\newcommand{\ee}{\end{equation}}
\newcommand{\beq}{\begin{eqnarray}\label}
\newcommand{\eeq}{\end{eqnarray}}
\newcommand{\bu}{\bi{u}}
\newcommand{\bx}{\bi{x}}
\newcommand{\lgl}{\left\langle}
\newcommand{\rgl}{\right\rangle_{T}}

\newcommand{\OT}{\Or\left(T^{-1}\right)}
\newcommand{\Sum}[1]{\sum_{#1=1}^\infty}
\newcommand{\scrF}{\vert\mathscr{F}\vert}
\newcommand{\phexp}{{\vphantom{*}}} 
\newcommand{\supj}{\sup_{j\geqslant 1}}
\newcommand{\kmax}{k_{\mathrm{max}}}

\newtheorem{thm}{Theorem}
\newtheorem{lemma}{Lemma}
\newtheorem{corollary}{Corollary}
\newtheorem{remark}{Remark}

\definecolor{darkgreen}{HTML}{008000}
\definecolor{orange}{HTML}{FF7F00}

\begin{document} 
\title[Shell models and the Navier--Stokes equations]{\large How close are shell models to the $3D$ Navier--Stokes equations?}

\author{Dario Vincenzi$^1$\footnote{Also Associate, International Centre for Theoretical
Sciences, Tata Institute of Fundamental Research, Bangalore 560089, India.} and John D Gibbon$^2$}
\address{$^1$ Universit\'e C\^ote d'Azur, CNRS, LJAD, 06100 Nice, France}
\address{$^2$ Department of Mathematics, Imperial College London, London SW7 2AZ, UK}
\eads{\mailto{\mailto{dario.vincenzi@univ-cotedazur.fr}, j.d.gibbon@imperial.ac.uk}}

\begin{abstract}
Shell models have found wide application in the study of hydrodynamic turbulence because they are easily solved numerically even at very large Reynolds numbers. Although bereft of spatial variation, they accurately reproduce the main statistical properties of fully-developed homogeneous and isotropic turbulence. Moreover, they enjoy regularity properties which still remain open for the three-dimensional (3$D$) Navier--Stokes equations (NSEs). The goal of this study is to make a rigorous comparison between shell models and the NSEs. It turns out that only the estimate of the mean energy dissipation rate is the same in both systems. The estimates of
the velocity and its higher-order derivatives display a weaker Reynolds number dependence for shell models than for the $3D$ NSEs. Indeed,
the velocity-derivative estimates for shell models are found to be equivalent to those corresponding to a velocity gradient averaged version of the $3D$ Navier--Stokes equations (VGA-NSEs), while the velocity estimates are even milder. Numerical simulations over a wide range of Reynolds numbers confirm the estimates for shell models.
\end{abstract}





\section{Introduction}

\label{sect:introd}

Three-dimensional (3$D$) incompressible turbulent flows are characterized by a cascade of kinetic energy from the length scales at which the flow is generated to the scales at which viscous dissipation becomes predominant \cite{uriel,SrAn97,Sr99,PD2004}. Kinetic energy is usually injected at large scale by a body forcing or the boundary conditions, and it is transferred at a constant rate to smaller  scales by nonlinear interactions between the Fourier modes of the velocity. 
The cascade is strongly dissipated when the viscous-dissipation range is reached.
The range between the forcing and viscous scales is known as inertial range and is characterized by a kinetic-energy spectrum of the form $E(k)\sim k^{-5/3}$,
where $k$ is the wavenumber.
As viscosity is reduced, the dissipation range shrinks, while the inertial range extends to smaller and smaller length scales. The consequence of this is that, in the limit of vanishing viscosity, the mean energy dissipation rate tends to a nonzero value \cite{sreenivasan,CRD2009}, which is known as the `dissipative anomaly'.
\par\smallskip
A mathematically rigorous description of the generation of small scales in turbulent $3D$ Navier--Stokes flows can be achieved by considering the $L^{2m}$-norms of the velocity derivatives for weak solutions of the Navier--Stokes equations (NSEs) \cite{bdgm93,gibbon12,gibbon19,gibbon20}. Given a velocity field $\bi{u}$ over a periodic cube $\mathscr{V}=[0,\,L]^3$, volume integrals and norms can be defined as
\bel\label{Hnmdef}
H_{n,m}(t) = \int_{\mathscr{V}} \vert\bnabla^{n}\bi{u}\vert^{2m}\,\rmd V\,.
\ee
The well known scaling property of the NSEs $\bu(\bx,\,t) \to \mu^{-1}\bu\left(\bx/\mu,\,t/\mu^{2}\right)$
suggests the definition of a doubly-labelled set of dimensionless, invariant quantities \cite{gibbon20}
\bel\label{Fnmdef}
F_{n,m}(t) = \nu^{-1} L^{1/\alpha_{m,m}} H_{n,m}^{1/2m}\,,
\ee
for $0 \leqslant n < \infty$ and $1 \leqslant m \leqslant \infty$, where $\nu$ is the kinematic viscosity and 
\bel\label{alphadef}
\alpha_{n,m} = \frac{2m}{2m(n+1)-3}\,. 
\ee
It was shown in \cite{gibbon19,gibbon20} that for $1~\leqslant~n < \infty$ and $1~\leqslant~m~\leqslant~\infty$, together with $n = 0$ for $3 < m~\leqslant~\infty$
\bel\label{Fnmest1}
\lgl F_{n,m}^{\alpha_{n,m}}\rgl \leqslant~c_{n,m} \Rey^{3} + \Or\left(T^{-1}\right)
\qquad (\Rey \gg  1)\,,
\ee
where $\Rey$ is the Reynolds number and the time average up to time $T>0$ is defined by
\bel\label{timav}
\lgl\cdot\rgl = T^{-1}\int_{0}^{T}\cdot\;\rmd t\,.
\ee
Moreover, the set of estimates in (\ref{Fnmest1}) encompasses all the known a priori bounds for weak solutions of the $3D$ NSEs equations and shows how these bounds arise naturally from scale invariance \cite{gibbon19,gibbon20}. In those references it has been shown that a hierarchy of spatially averaged length scales $\ell_{n,m}(t)$ can be constructed from the $F_{n,m}$ in the following manner\,:
\bel\label{elldef}
\left(L\ell_{n,m}^{-1}\right)^{n+1} := F_{n,m}\,.
\ee
Higher values of $n$ allow the detection of smaller scales, while higher values of $m$ account for stronger deviations from the mean. Using (\ref{Fnmest1}) and (\ref{elldef}), followed by a H\"older inequality, one finds that
\bel\label{ellest}
\lgl L\ell_{n,m}^{-1}\rgl \leqslant~c_{n,m}\Rey^{\frac{3}{(n+1)\alpha_{n,m}}} \qquad (\Rey \gg 1)\,.
\ee
As noted in \cite{gibbon19}, while the estimate for the first in the hierarchy is $\Rey^{3/4}$ and is consistent with the inverse  Kolmogorov length, the limit as $n,\,m\to\infty$ is finite and is consistent with the fact that viscosity ultimately dissipates the cascade of energy\footnote{Strictly speaking there is a limit to the value of $\Rey$ beyond which kinetic scales are reached and the NSEs become invalid.}. 
\par\smallskip
The $\Rey$ dependence of the moments of $\bnabla\bi{u}$ has been studied within the multifractal formalism \cite{nelkin90,bbpvv91,uriel} and in numerical simulations of both the $3D$ NSE \cite{ssy07} and the Burgers equation \cite{cfpr12}. The calculation of $H_{n,m}$ for high $\Rey$ and large values of $n$ and $m$ nonetheless requires large numerical simulations
(see \cite{dggkpv13} for the $n=1$ case). At high $\Rey$, indeed, the injection and viscous-dissipation scales are widely separated, and the cascade process activates a wide range of length scales\,: an empirical argument due to Landau and Lifschitz~\cite{landau} indicates that the number of degrees of freedom of a turbulent velocity field grows as $\Rey^ {9/4}$(see also~\cite{uriel}). For this reason, the direct numerical simulation of high-$\Rey$ flows have remained a great challenge~\cite{MM1998,celani07,IGK,DYS,RMK2012}. In order to study the properties of fully developed turbulence, simplified models have thus been introduced that retain some of the properties of the NSEs but are much more tractable both theoretically and numerically. Among these, shell models of the energy cascade have played a major role  \cite{uriel,bohr,biferale,ditlevsen}. They consist of a system of ordinary differential equations for a set of complex scalar variables which can be regarded loosely as the amplitudes of the Fourier components of the velocity field. The structure of the equations mimics that of the NSEs in Fourier space. The nonlinear part has a form that recalls the vortex-stretching term, but the interactions between the velocity variables are local. A linear small-scale dissipation and generally a forcing are also included. Because of their scalar nature, shell models are unable to provide information on the spatial structure of the velocity field. However, they successfully reproduce the statistical properties of space-averaged quantities, such as the kinetic-energy spectrum, the velocity structure functions or the viscous-dissipation rate in isotropic turbulence. From a mathematical point of view, stronger results have been proved for shell models than for the $3D$ NSEs\,: for instance, the global regularity of strong solutions and the existence of a finite-dimensional inertial manifold~\cite{clt06,clt07} (see~\cite{flandoli} for analogous results on stochastic shell models).
\par\smallskip
Questions still remain, however, over how close the mathematical results of shell-models are to those for the NSEs. Shell-models are bereft of spatial variation while the behaviour of solutions of the NSEs equations differ widely depending upon their dimension. Indeed, although the notion of velocity gradient as a spatio-temporal field in shell models is not available, it is easy to define the analogue of the volume integral of powers of the velocity derivatives. For instance, the shell-model analogues of enstrophy and helicity have been studied extensively \cite{bohr,biferale,ditlevsen}. Nevertheless, it is not clear where the exact correspondence lies. The goal of this paper is to investigate the analogue of $H_{n,m}$ for shell models, both mathematically and numerically, in relation to the NSEs equations to see if there is a consistent correspondence between the two.
\par\smallskip
The paper is organized in the following steps. \Sref{sect:definitions} introduces the shell model and the mathematical framework. We consider the `Sabra' model \cite{sabra}, but the results are general and, in particular, also hold for the GOY model \cite{gledzer,oy87}, the only difference being in the constants that appear in the estimates.
\par\smallskip
The starting point of our study is a bound for the mean energy-dissipation rate, which corresponds to the $m=n=1$ case. This is obtained in \sref{sect:energy} by adapting to shell models the methods used by Doering and Foias \cite{df02} for the NSEs (see also \cite{giorgio,ggpp18} for the application of the same methods to magnetohydrodynamics and binary mixtures). The estimate for the dissipation rate coincide, as expected, with that for the $3D$ NSEs.
\par\smallskip
In sections~\ref{sect:lis} and~\ref{sect:FGT}, we keep $m=1$ but move to general $n$, i.e. we study $H_{n}\equiv H_{n,1}$, the analogue of the $L^2$-norm of the $n$-th order derivative of the velocity. We first prove two differential relations connecting $H_{n}$ and $H_{n+1}$, in the spirit of the ``ladder'' relations available for the NSEs \cite{bdg91,dg95}. Following the strategy applied in \cite{CF88,dg95,FMRT}, we then use these relations together with the bound for the energy dissipation rate to prove the existence of absorbing balls for all $H_n$ and to estimate the time average $\lgl H_{n}^{\frac{2}{n+1}}\rgl$ in terms of $\Rey$. This latter result is the counterpart of a bound proved by Foias, Guillop\'e and Temam \cite{fgt} for weak solutions of the $3D$ NSEs. It is further extended to general $n$ and $m$ in \sref{sect:HNM}, where it is shown that, in terms of the $\alpha_{n,m}$ defined above in (\ref{alphadef}) and \eref{Fnmest1}, the shell-model equivalent is 
\bel\label{alphasmdef}
\alpha_{n,m} = \frac{4}{n+1}\,,
\ee
which is independent of $m$.  
\par\smallskip
The form of these bounds and, in particular, their insensitivity to $m$, raise the question of how close these results are to the NSEs in any dimension. 
Comparing \eref{alphasmdef} with \eref{alphadef}, we find that $\alpha_{n,m}$ is greater for shell models
than for the $3D$ NSEs for all $n>1$ and $1\leqslant m\leqslant 0$. Therefore, the $\Rey$ dependence of the high-order velocity derivatives differs in the two systems and, in shell models, is significantly weaker.
It is indeed discovered in \sref{sect:relax} that, as far as the velocity-derivative estimates are concerned, the real PDE-equivalent of the shell models considered here is not the full $3D$ NSEs themselves but a version of these that we have called the `velocity gradient averaged Navier--Stokes equations' (VGA-NSEs).
While less specific in its definition as intermittency in multi-fractal theories \cite{uriel}, intermittent events in solutions of the NSEs have the property that excursions in $\bnabla\bu$ depart strongly from its average $\|\bnabla\bu\|_{2}$, thereby implying that for very short periods of time
\bel\label{im1}
L^{3/2}\,\frac{\|\bnabla\bu\|_{\infty}}{\|\bnabla\bu\|_{2}} \gg 1\,,
\ee
whereas making the approximation
\bel\label{rnse1}
L^{3/2}\,\frac{\|\bnabla\bu\|_{\infty}}{\|\bnabla\bu\|_{2}} = 1\,,
\ee
has the effect of suppressing strong events in $\bnabla\bu$. 
The VGA-NSEs are obtained by using \eref{rnse1} in the differential inequalities for the NSEs.
In fact, it can be thought of in the following way\,: a Gagliardo--Nirenberg inequality shows that 
\bel\label{im2}
\frac{\|\bnabla\bu\|_{\infty}}{\|\bnabla\bu\|_{2}} \leqslant c_{n} \kappa_{n}^{3/2}\,,
\qquad\qquad
\kappa_{n} = \left(\frac{\|\bnabla^{n+1}\bu\|_{2}}{\|\bnabla\bu\|_{2}}\right)^{1/n}\,.
\ee
The wave-number $\kappa_{n}(t)$ behaves as a higher moment of the enstrophy spectrum and has a lower bound expressed as $L^{-1} \leq \kappa_{n}(t)$. (\ref{rnse1}) occurs when one uses only the minimum of the right-hand side of the inequality in (\ref{im2}). One of the main results of this paper is that for all $n \geqslant 1$ and $1 \leqslant m \leqslant \infty$, the bounds and the exponents in the various bounded time-averages of VGA-NSEs and the shell models are equivalent.

The velocity estimates for the shell models are even milder than those for the VGA-NSEs. Indeed, it is shown that for the VGA-NSEs
\bel\label{rel2}
\left<\|\bu\|_{m}^{2}\right>_{T} \leqslant c_{m}\,(\nu L^{-1})^{2}\Rey^{3} \qquad (m\geqslant 3)\,,
\ee
\noindent
whereas the shell-model analogues of $\left<\|\bu\|_{m}^{2}\right>_{T}$ scale as $\Rey^2$.
\par\smallskip
Finally, \sref{sect:conclusions} concludes the paper by summarizing the estimates for shell models and
comparing them with numerical simulations of the `Sabra' model over a wide range of values of $\Rey$.

\section{The `Sabra' shell model}\label{sect:definitions}

Shell models of turbulence describe the velocity field by means of a sequence of complex variables $u_j$, $j=1,2,3,\dots$, which represent its Fourier components. In the `Sabra' model \cite{sabra}, the variables $u_j$ satisfy the following equations\,:
\begin{equation}
  \dot{u}^{\vphantom{*}}_{j} =
  \rmi(a k^{\vphantom{*}}_{j+1} u^*_{j+1} u^{\vphantom{*}}_{j+2}
  +b k^{\vphantom{*}}_j u^{\vphantom{*}}_{j+1} u^*_{j-1}
  -c k^{\vphantom{*}}_{j-1} u^{\vphantom{*}}_{j-1} u^{\vphantom{*}}_{j-2})
  -\nu k_j^2 u^{\vphantom{*}}_j + f^{\vphantom{*}}_{j}\, ,
\label{eq:shell}
\end{equation}
where $k_j=k_0\lambda^j$ ($k_0>0$, $\lambda>1$) are logarithmically-spaced wave numbers, $\nu$ is the kinematic viscosity, and the $f_j$ are complex and represent the Fourier amplitudes of the forcing. The `boundary conditions' for the velocity variables are $u_0=u_{-1}=0$, while $k_0^{-1}$ plays the role of the largest spatial scale in the system. The coefficients $a$, $b$, $c$ are real and satisfy
\begin{equation}\label{eq:constraint}
a + b + c = 0\,.
\end{equation}
This condition ensures that the kinetic energy,
\begin{equation}
E=\Sum{j}\vert u_j\vert^{2}\,,
\end{equation}
is conserved when $\nu=0$ and $f_j=0$ for all $j$. In the inviscid, unforced case and under condition~\eref{eq:constraint},
the shell model also possesses a second quadratic invariant,
which for suitable values of $a$, $b$, $c$ can be interpreted as either a generalized helicity or a generalized enstrophy
\cite{bohr}. The parameters of the shell model can also be tuned so as to generate an inverse cascade of energy from small to large scales, as in two-dimensional turbulence \cite{gilbert}. In the following, however, we shall not impose any additional
constraint on $a$, $b$, $c$ other than~\eref{eq:constraint}.
\par\smallskip
Various forcings have been considered in the literature, such as those that act only on few low-$j$ shells and mimic the injection of energy at large scales \cite{oy87,sabra}, those that impose a constant energy input \cite{bcr00}, or those with a power-law `spectrum'~\cite{mov02,bclst04}.  We consider a constant-in-time deterministic forcing, but the results are easily generalized to time-dependent $f_j$. Following \cite{df02}, we define the forcing in a way as to isolate its magnitude from its shape. We take
\begin{equation}\label{eq:forcing}
  f_j=\mathscr{F}\phi_{j-j_f}\,,
\end{equation}
where $\mathscr{F}$ is a complex constant, $j_f\geqslant 1$, and the shape function $\phi_p$ is such that $\phi_p=0$ if $p<0$. Thus, $k_f=k_0 \lambda^{j_f}$ is the characteristic wavenumber of the forcing. We also assume
\begin{equation}\label{eq:l2-forcing}
\sum_{p=0}^\infty\vert \phi_p\vert^{2} < \infty
\end{equation}
and 
\begin{equation}\label{eq:normalization}
\sum_{p=0}^\infty \lambda^{-2p}\vert \phi_p\vert^2=1\,.
\end{equation}
The assumption in~\eref{eq:l2-forcing} means that the `energy' of the forcing is finite, while \eref{eq:normalization} is a normalization condition on the shape function. 
In sections~\ref{sect:HN} and~\ref{sect:relax}, we shall further require that the forcing has a maximum wavenumber $k_{\rm max}$.

Under assumption~\eref{eq:l2-forcing}, it was shown in \cite{clt06} that if the energy $E$ is bounded at time $t=0$, then it stays bounded at any later time. This allows us to define the root mean square velocity
\begin{equation}
U = \lgl E\rgl^{1/2},
\end{equation}
where the time average up to time $T$ has been introduced in \eref{timav}. In addition, the time-averaged dissipation rate
\begin{equation}\label{eq:def-epsilon}
\epsilon = \nu\lgl\Sum{j}k_j^2\vert u^{\vphantom{*}}_j\vert^2\rgl
\end{equation}
is also bounded for all $T>0$ \cite{clt06}. By using $U$, $k_f$, and $\vert\mathscr{F}\vert$, we can then define the Reynolds and Grashof numbers as 
\begin{equation}\label{Reydef}
  \Rey = \frac{U}{\nu k_f}\qquad\mbox{and}\qquad \Gr = \frac{\vert\mathscr{F}\vert}{\nu^2 k_f^3},
\end{equation}
respectively. The latter is a dimensionless measure of the forcing, while the former quantifies the response of the system.
\par\smallskip
Finally, we note that it is possible to introduce a suitable functional setting for the study of \eref{eq:shell}, in which a solution $\bi u=(u_1,u_2\dots)$ is regarded as an element of the sequences space $\ell^2$ over the field of complex numbers, with scalar product $(\bi{u},\bi{v})=\Sum{j}u^{\vphantom{*}}_j v_j^*$ for any $\bi{u},\,\bi{v} \in \ell^2$ \cite{clt06}. Here, however, we  follow the physical notation and work with the variables $u_j$ directly.

\section{The time-averaged energy dissipation rate}\label{sect:energy}

An estimate of the time-averaged energy dissipation rate $\epsilon$, defined in (\ref{eq:def-epsilon}), is an essential element of the present study as results on the high-order derivatives of the velocity are based on this. It was once conventional to write estimates for $\epsilon$ in terms of the Grashof number $\Gr$ until Doering and Foias~\cite{df02} introduced a method that converted these into estimates in terms of the Reynolds number $\Rey$, which is much more useful for comparison with other theories of turbulence. The methods used here are adapted from Doering and Foias~\cite{df02}.
\par\smallskip
Let us first introduce the constants that will appear in the estimate for $\epsilon$\,:
\begin{equation}
  \fl
  A = 
  \vert a\vert\lambda+\vert b\vert + \vert a+b\vert\lambda^{-1},
  \qquad
  B_\gamma = \sup_{p\geqslant 0}\lambda^{-(2\gamma-1)p}\vert\phi_p\vert,
  \qquad
  C_\gamma=\sum_{p=0}^\infty \lambda^{2\gamma p}\vert\phi_p\vert^2.
\end{equation}
$A$ is a function of the parameters of the shell model, while $B_\gamma$ and $C_{\gamma}$ are fixed by the shape of the forcing. The exponent $\gamma$ is a real number and must be such that $B_\gamma$ and $C_\gamma$ are finite;
in particular, the normalization condition in \eref{eq:normalization} implies $C_{-1}=1$. It is important to stress that $B_\gamma$ and $C_\gamma$ depend neither on the amplitude nor on the characteristic wavenumber of the forcing. We shall also make use of the following result\,: 
\begin{lemma}\label{lemma}
For any $\gamma\in\mathbb{R}$ such that $C_\gamma$ is finite,
\begin{equation}
      \Sum{j}k_j^{2\gamma}\vert f^{\vphantom{*}}_j\vert^2  = 
      C_{\gamma} k_f^{2\gamma} \scrF^{2}\,.
\end{equation}
\end{lemma}
\begin{proof} By using the definition of the forcing in \eref{eq:forcing} and rearranging the terms in the sum, we obtain\,:
\begin{eqnarray}
  \fl
  \Sum{j}k_j^{2\gamma}\vert f^{\vphantom{*}}_j\vert^2  
  & =  \scrF^2 \sum_{j=j_f}^{\infty}k_j^{2\gamma}
  \vert\phi^{\vphantom{*}}_{j-j_f}\vert^2
  = \scrF^2 \sum_{p=0}^{\infty}k_{p+j_f}^{2\gamma}
  \vert\phi^{\vphantom{*}}_{p}\vert^2
  \\
  & =  \scrF^2 \sum_{p=0}^{\infty}k_0^{2\gamma}
  \lambda^{2\gamma(p+j_f)}\vert\phi^{\vphantom{*}}_{p}\vert^2
  = \scrF^2 (k_0\lambda^{j_f})^{2\gamma}\sum_{p=0}^{\infty}\lambda^{2p\gamma}
  \vert\phi_{p}\vert^2.
  \label{eq:lemma-c_gamma}
\end{eqnarray}
Replacing the definitions of $k_f$ and $C_\gamma$ in \eref{eq:lemma-c_gamma} yields the result.
\end{proof}
\par\smallskip
As discussed above at the beginning of this section we now use the method of Doering and Foias~\cite{df02} to estimate the time-averaged dissipation rate $\epsilon$.
\begin{thm}\label{thm:epsilon}
Let the forcing $(f_1, f_2, \dots)$ be as in Sect.~\ref{sect:definitions} and the initial energy $E(0)$ be bounded. Then the time-averaged energy dissipation rate satisfies
\begin{equation}
\label{eq:ineq-Re}
\epsilon \leqslant \nu^3 k_f^4
\left(c_1{\Rey}^2+c_2 \Rey^3\right) + \Or(T^{-1})\,,
\end{equation}
where the constants
\begin{equation}
  c_{1} = \frac{\sqrt{C_0 C_{2(1-\gamma)}}}{C_{-\gamma}}\,,
  \qquad
  c_{2} = A\frac{\sqrt{C_0} B_{\gamma}}{C_{-\gamma}}\,,
\end{equation}
depend on the parameters $a$, $b$, $\lambda$ of the shell model and on the shape of the forcing $(\phi_1,\phi_2,\dots)$, but are uniform in $\nu$, $k_0$, $k_f$, $\scrF$.  The value of $\gamma\in\mathbb{R}$ may be chosen in a way as to minimize $c_1$ and $c_2$, but it must nonetheless be such that $C_{-\gamma}$, $C_{2(1-\gamma)}$, $B_{\gamma}$ are finite.
\end{thm}
\begin{remark}
The switch from $\Rey^{2}$ to $\Rey^{3}$ behaviour in (\ref{eq:ineq-Re}) is observed in the numerical computations displayed in Fig.~\ref{fig:1}.
\end{remark}
\begin{proof}
We begin by writing the evolution equation for the energy. Towards this end, we multiply \eref{eq:shell} by $u_j^*$ and the complex conjugate of \eref{eq:shell} by $u_j$. We then add the two resulting equations and sum over $j$\,:
\begin{equation}
\label{eq:energy}
\frac{\rmd E}{\rmd t}=-2\nu\Sum{j}k_j^2\vert u_j^\phexp\vert^{2} + \Sum{j}(f_j^\phexp u_j^*+f_j^* u_j^\phexp)\,.
\end{equation}
By integrating over time and using the Cauchy--Schwarz inequality twice on the last term, we obtain\,:
\begin{equation}
E(T) + 2\nu\int_0^T \left( \Sum{j}k_j^2\vert u_j^\phexp(t)\vert^2\right) \rmd t
\leqslant E(0)+2\sqrt{C_0}\,U\scrF\, T\,,
\end{equation}
whence
\begin{equation}
  \label{eq:epsilon-CS}
  \epsilon \leqslant \sqrt{C_0}\, U \scrF + \frac{E(0)}{2T}\,.
\end{equation}
To express this bound in terms of $\Rey$, we need to estimate $\scrF$ in terms of $U$ and $k_f$. We multiply \eref{eq:shell} by $k_j^{-2\gamma} f_j^*$, sum over $j$, and average over time\,:
\begin{eqnarray}
  \fl
  \lgl\Sum{j}k_j^{-2\gamma}f^*_j\dot{u}_j^\phexp\rgl=
  \Sum{j}k_j^{-2\gamma}\vert f_j^\phexp\vert^2  
  -\lgl\nu\Sum{j}k_j^{2-2\gamma}u_j^\phexp f_j^*\rgl
  \nonumber
  \\[1mm]
  +\lgl\rmi\Sum{j}k_j^{-2\gamma}f_j^* 
  \left(ak_{j+1}^\phexp u^*_{j+1}u_{j+2}^\phexp
  +b k_j^\phexp u_{j+1}^\phexp u^*_{j-1}
  -ck_{j-1}^\phexp u_{j-1}^\phexp u_{j-2}^\phexp \right)\rgl \,.
\label{eq:equality-F}
\end{eqnarray}
From \eref{eq:energy}, it is easy to see that $E(t)$ is bounded by a time-independent constant \cite{clt06}. This follows from using $k_j<k_1$ for all $j>1$ in the viscous term, the Cauchy--Schwarz inequality on the forcing term, and then Gronwall's inequality.
As a consequence, the left-hand side of \eref{eq:equality-F} is $\OT$.
\par\smallskip
The first term on the right-hand side is calculated from Lemma~\ref{lemma} as\,:
\begin{equation}\label{eq:term-1}
\Sum{j}k_j^{-2\gamma}\vert f_j\vert^2 = C_{-\gamma}^\phexp k_f^{-2\gamma} \scrF^{2}\,.
\end{equation}
The second term is estimated by using the Cauchy--Schwarz inequality\,: 
\begin{equation}\label{eq:term-2}
\left\vert \lgl\nu\Sum{j}(k_j^{2-2\gamma} f_j^*)u_j^\phexp\rgl \right\vert
\leqslant
\sqrt{C_{2-2\gamma}^\phexp}\, \nu U k_f^{2-2\gamma} \scrF\,.
\end{equation}
We estimate the third term by moving the forcing out of the sum and using again the Cauchy--Schwarz inequality\,:
\begin{eqnarray}
  \fl
  \left\vert
  \lgl\rmi\Sum{j}k_j^{-2\gamma}f_j^* 
  \left(a k_{j+1}^\phexp u^*_{j+1} u_{j+2}^\phexp
  +b k_j^\phexp u_{j+1}^\phexp u^*_{j-1} 
  -c k_{j-1}^\phexp u_{j-1}^\phexp u_{j-2}^\phexp \right)\rgl\right\vert
  &
  \nonumber
  \\[1mm]
  \fl
  \leqslant k_f^{-2\gamma+1}\scrF
  \left\vert\lgl
  \sum_{j=1}^\infty\lambda^{-(j-j_f)(2\gamma-1)}\phi_{j-j_f}
  \big(a\lambda u^*_{j+1} u_{j+2}^\phexp
  \right.\right. &
  \nonumber 
  \\*[-6mm]
  \left.\left.\vphantom{\Sum{j}}
  \qquad\qquad\qquad\qquad\qquad\quad
  +bu_{j+1}^\phexp u^*_{j-1} 
  +(a+b)\lambda^{-1} u_{j-1}^\phexp u_{j-2}^\phexp \big)
  \rgl\right\vert
  \nonumber
  \\[1mm]    
  \fl
  \leqslant
  A  B_{\gamma} U^2 k_f^{-2\gamma+1} \scrF\,.
  \label{eq:term-3}
\end{eqnarray}
We now combine \eref{eq:term-1} with \eref{eq:term-2} and \eref{eq:term-3} and find\,:
\begin{equation}\label{eq:ineq-F}
\scrF \leqslant \frac{\sqrt{C_{2-2\gamma}}}{C_{-\gamma}}\, \nu U k_f^{2} + 
A\frac{B_{\gamma}}{C_{-\gamma}}\, U^2 k_f +\OT \,.
\end{equation}
Inserting \eref{eq:ineq-F} into \eref{eq:epsilon-CS} and rearranging finally yields the estimate for $\epsilon$.
\end{proof}
The implications of \eref{eq:ineq-Re} for turbulent flows have been discussed thoroughly in \cite{df02} within the context of the $3D$ NSEs (see also \cite{dg02}). Here we briefly mention the shell-model counterpart of the main points\,:
\begin{enumerate}
\renewcommand{\theenumi}{\roman{enumi}}
\item 
The bound on $\epsilon$ can be rewritten as
\begin{equation}
\frac{\epsilon}{U^3 k_f}
\leqslant
\frac{c_1}{\Rey}+c_2 + \Or(T^{-1})\,.
\label{eq:epsilon-physical}
\end{equation}
Thus, in the high-$\Rey$ limit the saturation of the bound recovers the empirical prediction $\epsilon\sim U^3 k_f$ \cite{uriel}.
\item The estimate of $\epsilon$ can be converted into bounds for the Kolmogorov dissipation wavenumber $k_\eta=(\epsilon/\nu^3)^{1/4}$, the Taylor microscale $k_T=(\epsilon/\nu U^2)^{1/2}$, and the Taylor-microscale Reynolds number, $\Rey_\lambda=U/\nu k_T$. The saturation of these bounds for $\Rey\to\infty$ is consistent with the empirical predictions
$k_\eta\sim\Rey^{3/4}$, $k_T\sim\Rey^{1/2}$, $\Rey_\lambda\sim\Rey^{1/2}$ for $3D$ homogeneous and isotropic turbulence \cite{uriel}.
\item A lower bound for the time-averaged dissipation rate can also be derived by using the shell-model version of the Poincar\'e inequality\,:
\begin{equation}
\epsilon \geqslant
\nu k_1^2\lgl\Sum{j}\vert u_j\vert^2\rgl = \nu k_1^2 U^{2}\,,
\label{eq:dissipation}
\end{equation}
where we have used $k_1>k_j$ for all $j>1$. The latter bound can be rewritten as
\begin{equation}
\label{eq:epsilon-lower}
\frac{\epsilon}{U^3 k_f} \geqslant \left(\frac{k_1}{k_f}\right)^2 \Rey^{-1}\,.
\end{equation}
Therefore, the small-$\Rey$ scaling in \eref{eq:epsilon-physical} is sharp. Moreover, if we take $j_f=1$ and $\phi_p=\delta_{p,0}$, then $k_f=k_1$ and $c_1=1$. As a consequence, the upper and lower bounds on $\epsilon$ coincide for $\Rey\to 0$,
i.e. $\epsilon$ behaves as $\epsilon/U^3 k_f=\Rey^{-1}$. This means that the lower bound on $\epsilon$ is also optimal.
\item Dividing \eref{eq:ineq-F} by $\nu^2 k_f^3$ yields\,:
\begin{equation}\label{eq:Gr-Re}
  \Gr\leqslant c'_1 \Rey + c'_2 \Rey^{2}
\end{equation}
with $c'_1=c_1/\sqrt{C_0}$ and $c'_2=c_2/\sqrt{C_0}$. This bound establishes a relation between the forcing (represented by $\Gr$) and the response of the system (represented by $\Rey$).
\end{enumerate}
As mentioned earlier, the proof of Theorem~\ref{thm:epsilon} parallels that of Doering and Foias~\cite{df02} for the NSEs.
By using the same approach, it is possible to obtain estimates of $\epsilon$ in terms of $\Gr$ analogous to those available for the NSEs. It can indeed be shown that for $\Gr\to 0$ the lower and upper bounds on $\epsilon$ coincide, and hence
the time-averaged dissipation rate behaves as
$\epsilon=\nu^3 k_f^4 \Gr^2$, while for $\Gr\to\infty$ it satisfies the lower bound $\epsilon \geqslant c_3 \nu^3 k_1^2 k_f^2 \Gr$, where the constant $c_3$ is uniform in $\nu$, $k_0$, $\scrF$, and $k_f$. 
 
\section{High-order velocity derivatives}\label{sect:HN}

To investigate higher-order derivatives of the velocity, we now consider the sequence of infinite sums
\bel
H_n = \Sum{j} k_j^{2n} \vert u_j^\phexp\vert^2\,, \qquad n\geqslant 0\,,
\ee
which represent the shell-model analogues of the $L^2$-norms $\Vert\bnabla^n\bi{u}\Vert^2_{L^2}$.
Note that $H_0$ is the energy $E$, while the time average of $H_1$ is proportional to $\epsilon$\,:
\bel
\epsilon = \nu\lgl H_1\rgl\, .
\label{eq:epsilon-H_1}
\ee
We also denote the equivalent sums for the forcing variables as
\bel
\Phi_n=\Sum{j}k_j^{2n} \vert f_j\vert^2\,, \qquad n\geqslant 0\,.
\ee
Recall from Lemma~\ref{lemma} that $\Phi_n= C_n k_f^{2n} \scrF^2$\,.

\subsection{Ladder inequalities and absorbing balls for $H_n$}\label{sect:lis}

The following theorem shows that there exist two ladders of differential inequalities that connect 
$H_n$ and $H_{n+1}$ and reproduce the analogous ladder inequalities for the NSEs \cite{bdg91,dg95}.
\begin{thm} Let $n\geqslant 0$ and assume that the forcing $(f_1,f_2,\dots)$ is such that $\Phi_n<\infty$\,.
Then $H_n$ satisfies
\numparts
\begin{equation}\label{eq:ladder-grad}
\frac{1}{2}\dot{H}_n\leqslant -\nu H_{n+1}+c_n H_n  \supj k_{j}\vert u_{j}\vert + H_n^{\frac{1}{2}} \Phi_n^{\frac{1}{2}}
\end{equation}
and
\bel\label{eq:ladder-u}
\frac{1}{2}\dot{H}_n \leqslant -\frac{\nu}{2}\, H_{n+1} + \frac{d_n}{\nu}\, H_n  
\supj \vert u_{j}\vert^2 + H_n^{\frac{1}{2}} \Phi_n^{\frac{1}{2}} 
\ee
\endnumparts
with
\bel
c_n= \lambda^{-n+1}\left(\vert a\vert\lambda^{-2n}+\vert b\vert+\vert a+b\vert\lambda^{2n}\right)\, ,
\qquad 
d_n = \frac{c_n^2}{2\lambda^4}\, .
\label{eq:cn-dn}
\ee
\end{thm}
\begin{proof}
We multiply \eref{eq:shell} by $k_j^{2n}u_j^*$ and the complex conjugate of \eref{eq:shell} by $k_j^{2n}u_j$.
We then sum to obtain
\begin{eqnarray}
\label{eq:multiplication}
\fl
\dot{H}_n &=& -2\nu H_{n+1} 
\nonumber
\\
\fl
&& +\Sum{j} k_j^{2n} \left[\rmi u_j^* (a k_{j+1}^\phexp u^*_{j+1} u_{j+2}^\phexp
  +b k_j^\phexp u_{j+1}^\phexp u^*_{j-1} - c k_{j-1}^\phexp u_{j-1}^\phexp u_{j-2}^\phexp +f_j^\phexp)
  +{\rm c.c.}\right] \,, 
\end{eqnarray}
where `c.c.' stands for `complex conjugate'. The forcing term is estimated by using the Cauchy--Schwarz inequality\,:
\bel
\left\vert \Sum{j} k_j^{2n} u_j^* f_j^\phexp \right\vert =
\left\vert \Sum{j} (k_j^{n} u_j^*) (k_j^n f_j^\phexp) \right\vert
\leqslant H_n^{\frac{1}{2}} \Phi_n^\frac{1}{2}\,.
\label{eq:thm2-force}
\ee
Consider then the nonlinear term with coefficient $a$. We have 
\begin{eqnarray}
  \fl
  \left\vert a \Sum{n} k_j^{2n}k_{j+1}^\phexp u_j^* u^*_{j+1} u_{j+2}^\phexp \right\vert
  &\leqslant&
  \vert a\vert \lambda^{-3n+1} \supj \left(k_{j}\vert u_{j}\vert\right)
  \left\vert\Sum{j} (k_{j+1}^{n}u_{j+1}^*) (k_{j+2}^n u_{j+2}^\phexp)\right\vert
  \\
  &\leqslant& \vert a\vert \lambda^{-3n+1} H_n \supj \left(k_{j}\vert u_{j}\vert\right)  \,,
  \label{eq:thm2-grad}
\end{eqnarray}
where we have used $k_j^{n}=\lambda^{-np}k_{j+p}^{n}$ and the Cauchy--Schwarz inequality.
The terms with coefficients $b$ and $c=-(a+b)$ are treated in a similar manner.
The first ladder inequality is thus proved by using \eref{eq:thm2-force} and the estimates for the nonlinear terms in \eref{eq:multiplication}.

To prove \eref{eq:ladder-u}, we start again from~\eref{eq:multiplication}. The forcing term is estimated as above.
The term with coefficient $a$ is now manipulated as follows\,:
\begin{eqnarray}
    \fl
    \left\vert 
    a\Sum{j} k_j^{2n}k_{j+1}^\phexp u_j^* u^*_{j+1} u_{j+2}^\phexp \right\vert
    &\leqslant&
    \vert a\vert\lambda \supj \vert u_j\vert
    \left\vert\Sum{j}\left(k_j^{n} u^*_{j+1} \right) \left( k_j^{n+1}u_{j+2}^\phexp \right) \right\vert
    \\
    &\leqslant& 
    \vert a\vert \lambda^{-3n-1} H_{n\vphantom{+1}}^\frac{1}{2} H_{n+1}^\frac{1}{2}  \supj\vert u_j\vert \,,
\end{eqnarray}
where we have used the Cauchy--Schwarz inequality.
We then estimate the terms with coefficient $b$ and $c$ in a similar way and use Young's inequality to find
\begin{eqnarray}
\fl
\left\vert\Sum{j} k_j^{2n} u_j^* \left[(a k_{j+1}^\phexp u^*_{j+1} u_{j+2}^\phexp
  +b k_j^\phexp u_{j+1}^\phexp u^*_{j-1} - c k_{j-1}^\phexp u_{j-1}^\phexp u_{j-2}^\phexp)+\mathrm{c.c.}\right]
  \right\vert
\\
\leqslant 2 \sqrt{2d_n}\, H_{n\vphantom{+1}}^\frac{1}{2} H_{n+1}^\frac{1}{2} \supj\vert u_j\vert 
\leqslant \nu H_{n+1} + \frac{2d_n}{\nu}\, H_{n} \, \supj \vert u_j\vert^2 \,,
\label{eq:young}
\end{eqnarray}
where $d_n$ is defined in \eref{eq:cn-dn}.
Finally, we combine the first term on the right-hand side of \eref{eq:young}
with the viscous term in \eref{eq:multiplication} and add the estimate of the forcing term to get \eref{eq:ladder-u}.
\end{proof}

The structure of the ladder inequalities makes it evident that control over a low-$n$ rung of the ladder
automatically yields control over all the higher-order rungs \cite{bdg91,dg95}. 
Since $\supj \vert u_j\vert^2 \leqslant H_0$ and $H_0$ is bounded \cite{clt06},
inequality \eref{eq:ladder-force-2} can be used to prove that there are absorbing balls for all the $H_n$.
The existence of  absorbing balls was proved in \cite{clt06} by using different methods. Here we show how this result follows immediately
from the ladder inequalities and, in addition, 
we estimate the radius of the absorbing ball for $H_n$ under the assumption that $\Phi_n$ is finite.
\begin{corollary}
Let $n\geqslant 0$ and assume the forcing is such that $\Phi_n<\infty$, then
\bel
\limsup_{t\to\infty} H_n \leqslant \nu^2 k_f^{2(n+1)}\left[\tilde{d}_n\,\rho^{4(n+1)} Gr^{2(n+1)} + \widetilde{C}_n\,\rho^\frac{8}{n+2} Gr^2\right] \,,
\ee
where $\rho=k_f/k_1$ and 
\bel
\tilde{d}_{n} = 2^n d_n^n\,,  \quad \widetilde{C}_n =  2^\frac{2n}{n+2} C_n^\frac{n}{n+2}\,.
\ee
\end{corollary}
\begin{proof}
By using $\supj \vert u_j\vert^2\leqslant H_0$ and the inequality (see the Appendix for the proof)
\bel
\label{eq:triangle-ineq}
H_n \leqslant H_0^\frac{1}{n+1} H_{n+1}^\frac{n}{n+1} \,,
\ee
we rewrite \eref{eq:ladder-force-2} as
\bel
\dot{H}_n \leqslant - H_n \left[\nu\, \frac{H_n^\frac{1}{n}}{H_0^\frac{1}{n}}
- \frac{2d_n}{\nu}\, H_0  - 2\frac{\Phi_n^\frac{1}{2}}{H_n^\frac{1}{2}}\right] \,.
\ee
It follows that
\bel
\limsup_{t\to\infty} H_n \leqslant 2^n d_n^n \nu^{-2n} \limsup_{t\to\infty} H_0^{n+1}
+ 2^\frac{2n}{n+2} \nu^{-\frac{2n}{n+2}} \Phi_n^\frac{n}{n+2} \limsup_{t\to\infty} H_0^\frac{2}{n+2}\,.
\label{eq:limsup}
\ee
From Lemma~\ref{lemma}, we have
\bel
\label{eq:Phi}
\Phi_n= C_n k_f^{2n} \scrF^2 = C_n \nu^4 k_f^{2n+6} \Gr^2 \,.
\ee
In addition, it was shown in \cite{clt06} that
\bel
\label{eq:H0}
\limsup_{t\to\infty} H_0 \leqslant \nu^2 \left(\frac{k_f}{k_1}\right)^4 k_f^2 \Gr^2 \,.
\ee
Inserting \eref{eq:Phi} and \eref{eq:H0} into \eref{eq:limsup} yields the advertised result.
\end{proof}
It is also useful to reformulate the ladder inequalities in terms of the quantities
\bel
K_n = H_n + \tau^2 \Phi_n
\qquad \mathrm{with} \qquad
\tau = \nu^{-1} k_0^{-2}\,,
\ee
which incorporate the contribution of the forcing. This can be achieved under the additional assumption that the forcing has a cutoff in the spectrum, 
i.e.~there exists a maximum wavenumber $\kmax = k_0\lambda^{j_{\rm max}}$ such that $f_j=0$ for $j > 
j_{\mathrm{max}}$.
\begin{corollary}
If $n\geqslant 0$ and the forcing has a maximum wavenumber $\kmax$ and is such that $\Phi_n<\infty$, then $K_n$ satisfies
\numparts
\bel
\frac{1}{2}\,\dot{K}_n \leqslant -\nu K_{n+1} +c_n K_n  \supj k_{j}\vert u_{j}\vert + \nu\left(k_0^2+\kmax^2\right)K_n
\label{eq:ladder-force-1}
\ee 
and
\bel\label{eq:ladder-force-2}
\frac{1}{2}\,\dot{K}_n \leqslant -\frac{\nu}{2}\, K_{n+1} + \frac{d_n}{\nu}\, K_n  \supj\vert u_{j}\vert^2 + \nu\left(k_0^2+\kmax^2\right)K_n\,.
\ee 
\endnumparts
\end{corollary}
\begin{proof}
The strategy for deriving \eref{eq:ladder-grad} from \eref{eq:ladder-force-1} is the same as for the NSEs \cite{bdgm93,dg95}.
Note first that $\dot{H}_n=\dot{K}_n$. Then, add and substract $\nu\tau^2\Phi_{n+1}$ to the right-hand side of \eref{eq:ladder-grad} to obtain the negative definite term in \eref{eq:ladder-force-1}.  The remaining two terms of the $H_n$ inequality are expressed in terms of $K_n$ via the obvious bounds $H_n\leqslant K_n$ and $\Phi_n \leqslant \tau^{-2}K_n$. Finally, we are left with the term $\nu\tau^2\Phi_{n+1}$, which is  estimated by using $\tau^2\Phi_{n+1}\leqslant \Phi_{n+1}K_n/\Phi_n\leqslant \kmax^2 K_n$.

Inequality \eref{eq:ladder-force-2} is proved in exactly the same manner.
\end{proof}

\subsection{A bound for the time average $\lgl H_n^{\frac{2}{n+1}}\rgl$}
\label{sect:FGT}

We now make use of the first ladder inequality and the estimate for $\epsilon$ to prove the shell-model analogue of a Navier--Stokes result of Foias, Guillop\'e and Temam \cite{fgt}. It ought to be noted that the exponent of $H_n$ in the bound below is greater than that found for the $3D$ NSEs. The reason for this difference between the shell model and the $3D$ NSEs is discussed in Sect.~\ref{sect:relax}.
\begin{thm}\label{thm:FGT}
Let $n\geqslant 1$ and $E(0)<\infty$ and assume that the forcing $(f_1,f_2,\dots)$ has a maximum wavenumber $\kmax$ and is such that
$\Phi_n<\infty$. Then, for $\Rey\gg 1,$
\bel
\lgl H_n^{\frac{2}{n+1}}\rgl \leqslant \hat{c}_n \, \nu^\frac{4}{n+1} k_f^4\, \Rey^3  + \OT \,,
\label{eq:FGT}
\ee
where the dimensionless positive constant $\hat{c}_n$ depends on $a$, $b$, $\lambda$, $n$ but is uniform in $\nu$, $k_0$, $k_f$, $\kmax$, $\scrF$.
\end{thm}
\begin{proof}
By noting that
\begin{equation}
\supj k_j \vert u_j\vert = \left(\supj k_j^2 \vert u_j\vert^2\right)^{1/2} \leqslant H_1^{1/2}\,,
\end{equation}
we turn \eref{eq:ladder-force-1} into
\begin{equation}
\label{eq:ladder-thm-H_n}
\frac{1}{2}\,\dot{K}_n \leqslant -\nu K_{n+1} +c_n \widehat{H}_1^{1/2} K_n \,,
\end{equation}
where we have denoted $\widehat{H}_1^{1/2} = H_1^{1/2} + 2\nu\kmax^2$ and have used $k_0<\kmax$.
We shall see that the additive constant in $\widehat{H}_1^{1/2}$ gives a negligible contribution at large $\Rey$.

We then divide \eref{eq:ladder-thm-H_n} by $K_n^\frac{n}{n+1}$ and time average.
The time-derivative term can be simplified as follows\,:
\bel
\fl
\lgl K_n^{-\frac{n}{n+1}}\dot{K}_n\rgl = (n+1) \lgl \frac{\rmd}{\rmd t} K_n^{\frac{1}{n+1}}\rgl = 
\frac{n+1}{T}\left[K_n^{\frac{1}{n+1}}(T) - K_n^{\frac{1}{n+1}}(0)\right] \,.
\ee
The first term on the right-hand side is bounded below by $(n+1)\left(\tau^2\Phi_n\right)^\frac{1}{n+1}/T>0$, while the second 
one is $\OT$. We are therefore left with
\bel
\fl
\lgl \frac{K_{n+1}}{K_n^\frac{n}{n+1}}\rgl
\leqslant
\frac{c_n}{\nu}\lgl K_{n}^{\frac{1}{n+1}}\widehat{H}_1^{\frac{1}{2}}\rgl + \OT
\leqslant
\frac{c_n}{\nu}\lgl K_{n}^{\frac{2}{n+1}}\rgl^{\frac{1}{2}}
\lgl \widehat{H}_1\rgl^{\frac{1}{2}} + \OT \,.
\label{eq:average-ratio}
\end{equation}
We now estimate the time average of $K_{n+1}^{\frac{2}{n+2}}$ by using \eref{eq:average-ratio}
and H\"older's inequality\,:
\begin{eqnarray}
\fl
\lgl K_{n+1}^{\frac{2}{n+2}}\rgl & = &
\lgl \left(\frac{K_{n+1}}{K_n^\frac{n}{n+1}}\right)^{\frac{2}{n+2}}
K_n^\frac{2n}{(n+1)(n+2)}\rgl
\leqslant
\lgl \frac{K_{n+1}}{K_n^\frac{n}{n+1}}\rgl^{\frac{2}{n+2}}
\lgl K_n^\frac{2}{(n+1)}\rgl^\frac{n}{n+2}
\\
&\leqslant&
c'_n \nu^{-\frac{2}{n+2}}
\lgl K_n^\frac{2}{(n+1)}\rgl^\frac{n+1}{n+2}
\lgl \widehat{H}_1 \rgl^\frac{1}{n+2} + \OT
\label{eq:bound-Hn1}
\end{eqnarray}
with $c'_n = c_n^{\frac{2}{n+2}}$.
Define now the dimensionless quantities
\begin{equation}
{A}_1= \nu^{-2} k_f^{-4}\lgl \widehat{H}_1\rgl
\qquad \mathrm{and}\qquad
A_n= \nu^{-\frac{4}{n+1}} k_f^{-4} \lgl K_n^\frac{2}{(n+1)}\rgl
\end{equation}
for $n\geqslant 2$.
The bound in \eref{eq:bound-Hn1} then takes the form
\begin{equation}
\label{eq:bound-A}
\lgl A_{n+1}\rgl \leqslant c'_n  \lgl A_n\rgl^\frac{n+1}{n+2} \lgl {A}_1 \rgl^\frac{1}{n+2} +\OT
\end{equation}
and, after the use of Young's inequality,
\begin{equation}
\label{eq:Young}
\lgl A_{n+1}\rgl \leqslant \frac{c'_n(n+1)}{n+2} \lgl A_n \rgl + \frac{c'_n}{n+2} \lgl {A}_1 \rgl +\OT\,.
\end{equation}
To estimate ${A}_1$, we invoke Jensen's inequality, \eref{eq:epsilon-H_1}, and Theorem~\ref{thm:epsilon} for $\Rey\gg 1$\,:
\bel
\fl
\lgl {A}_{1}\rgl \leqslant \nu^{-2}k_f^{-4} \lgl H_1\rgl + 4 \nu^{-1}k_f^{-4}\kmax^2 \lgl H_1\rgl^{1/2} + 4k_f^{-4}\kmax^4
\leqslant \hat{c}_1\, \Rey^3 + \OT\,.
\label{eq:bound-A1}
\ee
Here $\hat{c}_1$ is a dimensionless constant that depends on $a$, $b$, $\lambda$ and is uniform in $\nu$, $k_0$, $k_f$, $\kmax$, $\scrF$.
We now use \eref{eq:bound-A1} in \eref{eq:Young} for $n= 1$ to estimate $\lgl A_2\rgl$ and then proceed iteratively to find
\begin{equation}
\lgl A_{n}\rgl
\leqslant \hat{c}_{n} \Rey^3 + \OT\, .
\end{equation}
We obtain the final result by writing the latter bound in dimensional form and recalling that $H_n\leqslant K_n$.
\end{proof}

\subsection{High-order moments of the velocity derivatives}\label{sect:HNM}

For Navier--Stokes flows, the deviations of the velocity and its derivatives from their mean values are captured by the norms $\Vert\bnabla^n\bi{u}\Vert_{L^{2m}}$, 
where $0\leqslant n$ and $1\leqslant m\leqslant \infty$ \cite{gibbon20}.
For $m<\infty$, the shell-model analogues of $\Vert\bnabla^n\bi{u}\Vert^{2m}_{L^{2m}}$ are
\bel
H_{n,m} = \Sum{j} k_j^{2nm} \vert u_j^\phexp\vert^{2m} \,,
\ee
which reduce to $H_n$ when $m=1$.
Instead, the analogue of $\Vert\bnabla^n\bi{u}\Vert_{L^\infty}$ is $\supj k_j^n \vert u_j^\phexp \vert$.

By building on the results of the previous sections, we can generalize Theorem~\ref{thm:FGT} to $H_{n,m}$. Note once again that the exponent of $H_{n,m}$ in the time average differs
from that found for weak solutions of the $3D$ NSEs \cite{gibbon20}.
\begin{thm}\label{thm:HNM}
Under the same assumptions as in Theorem~\ref{thm:FGT} and for $\Rey\gg 1$, $H_{n,m}$ satisfies
\begin{equation}
\lgl H_{n,m}^{\frac{2}{m(n+1)}}\rgl \leqslant \hat{c}_{n}\, \nu^{\frac{4}{(n+1)}} k_f^4 \, \Rey^3 + \OT
\label{eq:HNM}
\end{equation}
if $1\leqslant n$, $1\leqslant m < \infty$, and
\begin{equation}
\lgl H_{0,m}^{\frac{1}{m}}\rgl \leqslant \nu^2 k_f^2 \, \Rey^2
\label{eq:H0M}
\end{equation}
if $n=0$ and $1\leqslant m < \infty$. In addition, for $n\geqslant 1$
\bel\label{eq:infty-norm}
\lgl \Big(\supj k_j^n \vert u_j^\phexp \vert\Big)^\frac{4}{n+1} \rgl \leqslant \hat{c}_{n}\, \nu^{\frac{4}{(n+1)}} k_f^4 \, \Rey^3 + \OT\,,
\ee
while for $n=0$
\bel
\lgl \Big(\supj \vert u_j^\phexp \vert\Big)^{2} \rgl \leqslant \nu^2 k_f^2 \, \Rey^2.
\label{eq:infty-norm-n_0}
\ee
The constants $\hat{c}_{n}$ depend on $a$, $b$, $\lambda$, $n$, but are uniform in $\nu$, $k_0$, $k_f$, $\kmax$, $\scrF$.
\end{thm}
\begin{proof}
The case $m=1$ was proved in Theorem~\ref{thm:FGT}.
For $1<m<\infty$, we use the inequality
\bel
\Sum{j} X_j \leqslant \left(\Sum{j} X_j^{1/p}\right)^p\,,
\label{eq:ineq-power}
\ee
where $p\geqslant 1$ and $X_j\geqslant 0$ for all $j$. When applied to $H_{n,m}$, this inequality yields
\bel
H_{n,m} \leqslant H_{n}^{m}\,.
\label{eq:LHNM-HN}
\ee
If $n\geqslant 1$, the result follows from raising both sides of \eref{eq:LHNM-HN} to the power $2/m(n+1)$ and invoking Theorem~\ref{thm:FGT}. For $n=0$, it is proved by raising both sides of \eref{eq:LHNM-HN} to the power $1/m$ and using $H_0=\nu^2 k_f^2 \Rey$.

Finally, \eref{eq:infty-norm} is proved by noting that
\bel
\left(\supj k_j^n \vert u_j^\phexp \vert\right)^\frac{4}{n+1} = \left(\supj k_j^{2n} \vert u_j^\phexp \vert^2\right)^\frac{2}{n+1}
\leqslant H_n^\frac{2}{n+1}
\ee
and using Theorem~\ref{thm:FGT}, while \eref{eq:infty-norm-n_0} follows from $\supj \vert u_j\vert^2\leqslant H_0$.
\end{proof}

\section{Comparison with the velocity gradient averaged Navier--Stokes equations}\label{sect:relax}

The issue in this section concerns how the velocity derivative estimates displayed in Theorem \ref{thm:HNM} compare with those for the NSEs. It is not clear that there necessarily {\em should} be a positive comparison, given that the $3D$ NSEs are not known to be regular and their corresponding scaling exponents
defined in \eref{Fnmdef} and \eref{Fnmest1} are different, namely\,:
\bel\label{comp1}
\alpha_{n,m} = \frac{2m}{2m(n+1)-3}~~\mbox{(NSE)}\qquad\quad \alpha_{n,m} = \frac{4}{n+1}~~\mbox{(Shell)}\,.
\ee
As we will now show, the real comparison lies with what we have called the ``velocity gradient averaged Navier--Stokes equations'' (VGA-NSEs). To explain the origin of this name, let us return to the first ladder inequality for $H_{n}$ displayed in (\ref{eq:ladder-grad}), which for the NSEs is written in the form\footnote{In this section, $c_n$ is a generic positive constant dependent on $n$.}
\bel\label{eq:ladder-grad-ex}
\frac{1}{2} \dot{H}_{n} \leqslant - \nu H_{n+1} +  c_{n}\|\bnabla\bu\|_{\infty}H_{n}\,,
\ee
where for simplicity, we have ignored the forcing term \cite{bdg91,dg95}. 
As explained in (\ref{im1}) in \S\ref{sect:introd}, the approximation where the $L^{\infty}$-norm is replaced by its spatial average
\bel\label{rnse1a}
\|\bnabla\bu\|_{\infty}\leqslant c_{n} L^{-3/2}\|\bnabla\bu\|_{2}\,.
\ee
has the effect of suppressing intermittent events in $\bnabla\bu$.
Thus we are not dealing with a modified PDE but with an averaging of its solutions reflected in the behaviour of $\|\bnabla\bu\|_{\infty}$. In terms of the $H_{n}$-ladder we are dealing with 
\bel\label{rnse2}
\frac{1}{2} \dot{H}_{n} \leqslant - \nu H_{n+1} +  c_{n} L^{-3/2}\|\bnabla\bu\|_{2}H_{n}
\ee
which yields the exact equivalent of Theorem \ref{thm:FGT}\,:
\begin{thm}\label{FGT-RNSE}
For $n \geqslant 1$, the $H_{n}$ for the $3D$ VGA-NSEs obey the bounds
\bel\label{rnse3}
\left< H_{n}^{\frac{2}{n+1}}\right>_{T} \leqslant c_{n} L^{-4}\nu^{\frac{4}{n+1}}\,\Rey^{3}\,.
\ee
\end{thm}
\begin{remark}
Bounds for $H_{n,m}$ follow in the same manner as in Theorem \ref{thm:HNM}, as can be easily seen by using approximation \eref{rnse1} in the proof of Theorem~1 of \cite{gibbon19}. The relaxation of the $L^{\infty}$ to the $L^{2}$-norm in \eref{rnse1} accounts for the insensitivity of the exponents to the value of $m$.
\end{remark}
\begin{proof}
To mimic the FGT-analysis of Theorem \ref{thm:FGT}, and suppressing the multiplicative factors of $L$ and $\nu$, we divide (\ref{rnse2}) by $H_{n}^{\frac{n}{n+1}}$ to obtain
\beq{nse5}
\left<\frac{H_{n+1}}{H_{n}^{\frac{n}{n+1}}}\right>_{T} &\leqslant& \left<H_{1}^{1/2}H_{n}^{\frac{1}{n+1}}\right>_{T}\nonumber\\
&\leqslant& \left<H_{1}\right>_{T}^{1/2}\left<H_{n}^{\frac{2}{n+1}}\right>_{T}^{1/2}
\eeq
Moreover, 
\beq{nse6}
\left<H_{n+1}^{\frac{2}{n+2}}\right>_{T} &=& \left<\left(\frac{H_{n+1}}{H_{n}^{\frac{n}{n+1}}}\right)^{\frac{2}{n+2}}
H_{n}^{\frac{2n}{(n+1)(n+2)}}\right>_{T}\nonumber\\
&\leqslant& \left<\frac{H_{n+1}}{H_{n}^{\frac{n}{n+1}}}\right>_{T}^{\frac{2}{n+2}}
\left<H_{n}^{\frac{2}{n+1}}\right>_{T}^{\frac{n}{n+2}}
\eeq
Let 
\bel\label{Xndef}
X_{n} = 
\left<H_{n}^{\frac{2}{n+1}}\right>_{T}
\ee
then from (\ref{nse6}) and (\ref{nse5}) we have 
\beq{nse7}
X_{n+1} &\leqslant& \frac{2}{n+2}\left<\frac{H_{n+1}}{H_{n}^{\frac{n}{n+1}}}\right>_{T} + \frac{n}{n+2}X_{n}\nonumber\\
&\leqslant&  \frac{1}{n+2}\left(\left<H_{1}\right>_{T} + X_{n}\right) + \frac{n}{n+2}X_{n}\nonumber\\
&=&  \frac{1}{n+2}\left<H_{1}\right>_{T} + \frac{n+1}{n+2}X_{n}
\eeq
Since $X_{1} = \left<H_{1}\right>_{T} \leqslant Re^{3}$ we have estimates for every $n \geqslant 1$ in the form of (\ref{rnse3}).
\end{proof}
\noindent
Theorem \ref{FGT-RNSE} holds for $n \geqslant 1$. What of the velocity field represented by $n=0$?  
\begin{lemma}
For  $3 < m \leqslant \infty$, the velocity field for the $3D$ VGA-NSEs obey the bounds 
\bel\label{u2}
\left<\|\bu\|_{2m}^{2}\right>_{T} \leqslant c_{m}\Rey^{3}\,.
\ee
\end{lemma}
\begin{remark}
The problem lies in understanding what is happening in the range $1 \leqslant m \leqslant 3$, which remains an open problem.
\end{remark}
\begin{proof}
Firstly we note that
\bel\label{gni1}
\|\bu\|_{2m} \leqslant c\, \|\bnabla\bu\|_{\infty}^{a}\|\bu\|_{p}^{1-a}
\ee
for $p > 2m$ and $a= \frac{3(p-2m)}{p(3+2m)}$. Moreover, 
\bel\label{gni2}
\|\bu\|_{p} \leqslant c\, \|\bnabla\bu\|_{\infty}^{A}\|\bu\|_{p}^{1-A}
\ee
where $A=\frac{3(p-2m)}{p-6}$, $p > 6$, $m > 3$. Then the $L^{\infty}\to L^{2}$ replacement as in \eref{rnse1} gives 
\bel\label{gni3}
\|\bu\|_{2m} \leqslant c\,L^{-\frac{m-3}{2m}}\|\bnabla\bu\|_{2}\qquad\mbox{for}\qquad m > 3\,.
\ee
This is exactly a `less intermittent' form of Sobolev's inequality which allows some variation in the $L^{2m}$-norm on the left-hand side instead of $L^{6}$ alone, as in its standard form.
\end{proof} 
Comparing \eref{eq:H0M} with \eref{u2} shows that the equivalence between the shell model and the VGA-NSEs only holds at the level of the velocity derivatives. In shell models, the dependence of the velocity field on $\Rey$ is even weaker than in the VGA-NSEs.

\section{Simulations and concluding remarks}\label{sect:conclusions}

To test the mathematical estimates, we have performed numerical simulations of the Sabra model. The parameters are the typical ones used in studies of $3D$ turbulence\,: $a=1$, $b=c=-1/2$, $k_0=2^{-4}$, $\lambda=2$ \cite{oy87}.
The forcing has the form $f_j=\mathscr{F}\delta_{j,1}$ with $\mathscr{F}=5\times 10^{-3}(1+\rmi)$, and the viscosity is varied between $\nu=10^{-7}$ and $\nu \approx 6\times 10^2$. We truncate the system to $N$ shells by imposing the additional boundary conditions $u_{N+1}=u_{N+2}=0$, where $N$ is varied between 8 and 27 depending on the value of $\nu$.  The numerical integration uses a second-order slaved Adams--Bashforth scheme \cite{pisarenko} with time step $\rmd t= 10^{-4}$.

Figures \ref{fig:1} to \ref{fig:3} show $\epsilon$, $\Gr$, and $\Big\langle H_{n,m}^\frac{2}{m(n+1)}\Big\rangle_T$ for different values of $n$ and $m$ as a function of $\Rey$. To facilitate the reading of the figures, the relevant definitions and estimates are summarized in Table~\ref{table}. The values of $\Rey$ vary from the `laminar' regime, in which the shell model relaxes to a fixed point, to the fully turbulent regime, which is characterized by a $k_j^{-5/3}$ spectrum over several decades of wavenumbers. 

\begin{table}
\small
\centering
\begin{tabular}{l c l c}
\hline
Definition && Estimate & Reference
\\
\hline
\\
$\epsilon = \nu \lgl \Sum{j} k_j^2 \vert u_j^\phexp\vert^2 \rgl$ && 
\parbox{4.6cm}{$\epsilon \leqslant \nu^3 k_f^4\left(c_1{\Rey}^2+c_2 \Rey^3\right)$\\[2mm]
$\epsilon\, U^{-3}k_f^{-1} \leqslant c_1\Rey^{-1} + c_2$}
& \parbox{0.7cm}{\eref{eq:ineq-Re} \\[2mm] \eref{eq:epsilon-physical}}
\\[9mm]
$\Gr = \scrF / \nu^2 k_f^3$ && $\Gr \leqslant c_1' \Rey + c'_2 \Rey^2$ & \eref{eq:Gr-Re}
\\[5mm]
$H_{n}=\Sum{j} k_j^{2n}\vert u_j^\phexp\vert^2$ && $\Big\langle H_n^{\frac{2}{n+1}}\Big\rangle_T \leqslant \hat{c}_n \, \nu^\frac{4}{n+1} k_f^4\, \Rey^3 \quad (n\geqslant 1,\,\Rey\gg 1)$ & \eref{eq:FGT}
\\[7mm]
$H_{n,m}=\Sum{j} k_j^{2nm}\vert u_j^\phexp\vert^{2m}$ && $\Big\langle H_{n,m}^{\frac{2}{m(n+1)}}\Big\rangle_T \leqslant \hat{c}_{n}\, \nu^{\frac{4}{(n+1)}} k_f^4 \, \Rey^3
\quad (n\geqslant 1,\, \Rey\gg 1)$ & \eref{eq:HNM}
\\[7mm]
$H_{0,m}=\Sum{j}\vert u_j^\phexp\vert^{2m}$ && $\lgl H_{0,m}^{1/m}\rgl\leqslant \nu^2 k_f^2 \Rey^2$  & \eref{eq:H0M}
\\[7mm]
&& $\lgl \Big(\supj k_j^n \vert u_j^\phexp \vert\Big)^\frac{4}{n+1} \rgl \leqslant \hat{c}_{n}\, \nu^{\frac{4}{(n+1)}} k_f^4 \, \Rey^3 \quad (n\geqslant 1)$  & \eref{eq:infty-norm}
\\[7mm]
&& $\lgl \Big(\supj \vert u_j^\phexp \vert\Big)^{2} \rgl \leqslant \nu^2 k_f^2 \, \Rey^2$ & \eref{eq:infty-norm-n_0}
\\
\\
\hline
\end{tabular}
\caption{Summary of the main estimates and definitions. The $\OT$ corrections have not been included for simplicity.}
\label{table}
\end{table}

The simulations clearly show that the mathematical estimates in Table~\ref{table} accurately describe the behaviour of the shell model as a function of \textit{Re}. Figure~\ref{fig:2}(b) also indicate that, for $\Rey\ll 1$, the scaling of $\Big\langle H_{n}^\frac{2}{(n+1)}\Big\rangle_T$ depends on $n$, as may be inferred from the proof of Theorem~\ref{thm:FGT} (see \eref{eq:bound-A} to \eref{eq:bound-A1}). Related to this, in Fig.~\ref{fig:3}(a) the small-$\Rey$ scaling of 
$\Big\langle H_{n,m}^\frac{2}{m(n+1)}\Big\rangle_{T}$ depends on $n$ but not on $m$, as a consequence of $H_{n,m}$ being controlled by $H_{n}^m$ (see \eref{eq:LHNM-HN}).
\par\medskip
Our conclusion is that shell models behave more closely to the $3D$ VGA-NSEs than the NSEs themselves. They both have identical scaling exponents in their time averages of their velocity derivatives which are reflected in the suppression of strong events of $\bnabla\bi{u}$, as proposed in equation (\ref{rnse1}). The actual properties of shell models for the velocity field itself are even milder than the estimates for the VGA-NSEs\,: compare (\ref{eq:H0M}) in Table 1 with \eref{u2}.
\par\medskip

Finally, we ask how much more regularity do solutions of shell models possess than those for the NSEs? This is shown up by comparing the estimates for velocity derivatives. Consider the scaling exponents $\alpha_{n,m}$ defined in (\ref{alphadef}) which appear in (\ref{Fnmest1}). It is not difficult to replicate this result in $D= 3,\,2,\,1$ dimensions \cite{gibbon20}. 
$\alpha_{n,m}$ is replaced by $\alpha_{n,m,D}$ 
\bel\label{alphaDdef}
\alpha_{n,m,D} = \frac{2m}{2m(n+1)-D}
\ee
and the relation involving $F_{n,m}$ in (\ref{Fnmest1}) is replaced by
\bel\label{D0a}
\left< F_{n,m,D}^{(4-D)\alpha_{n,m,D}}\right>_{T} \leqslant c_{n,m,D}\Rey^{3}\,.
\ee
In all these estimates, the larger the exponent the more regularity we have. Under what conditions is the $4/(n+1)$ of shell models greater than $(4-D)\alpha_{n,m,D}$? 
\bel\label{D0b}
\frac{4}{n+1} \geq (4-D)\alpha_{n,m,D}\,? 
\ee
The answer turns out to be 
\bel\label{D0c}
2D\left\{m(n+1) -2\right\} \geqslant 0\,,
\ee
and is thus always true when $n \geqslant 1$ and $m \geqslant 1$ for every value of $D$. Equality holds only at the level of the energy dissipation rate when $n=m=1$. The same result implies that, exception made for the time-averaged dissipation rate, the $\Rey$-dependence of the velocity derivatives is weaker for shell models than for the $D$-dimensional NSEs for any integer $D$.
Curiously, in a formal manner, equality also holds in the limit $D\to 0$, which corresponds to the ``Navier--Stokes equations on a point'', which has zero dimension. Given that shell models have no spatial variation the physical correspondence between the two is intriguing.

\begin{figure}
\setlength{\unitlength}{\textwidth}
\includegraphics[width=0.508\textwidth]{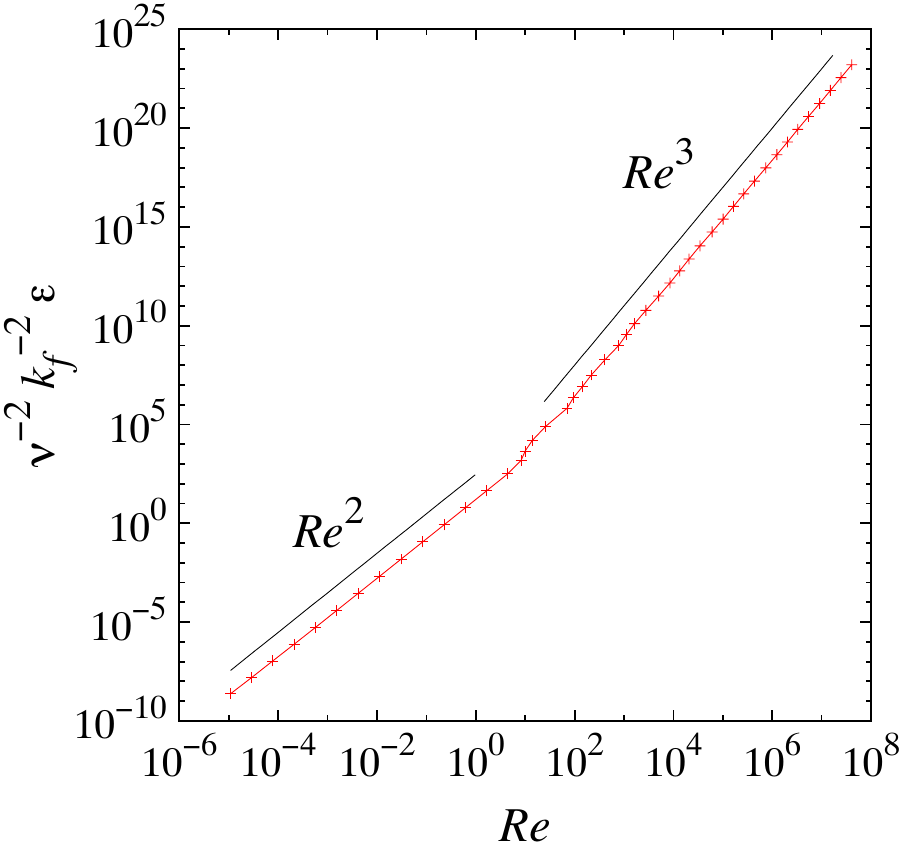}\hfill%
\includegraphics[width=0.492\textwidth]{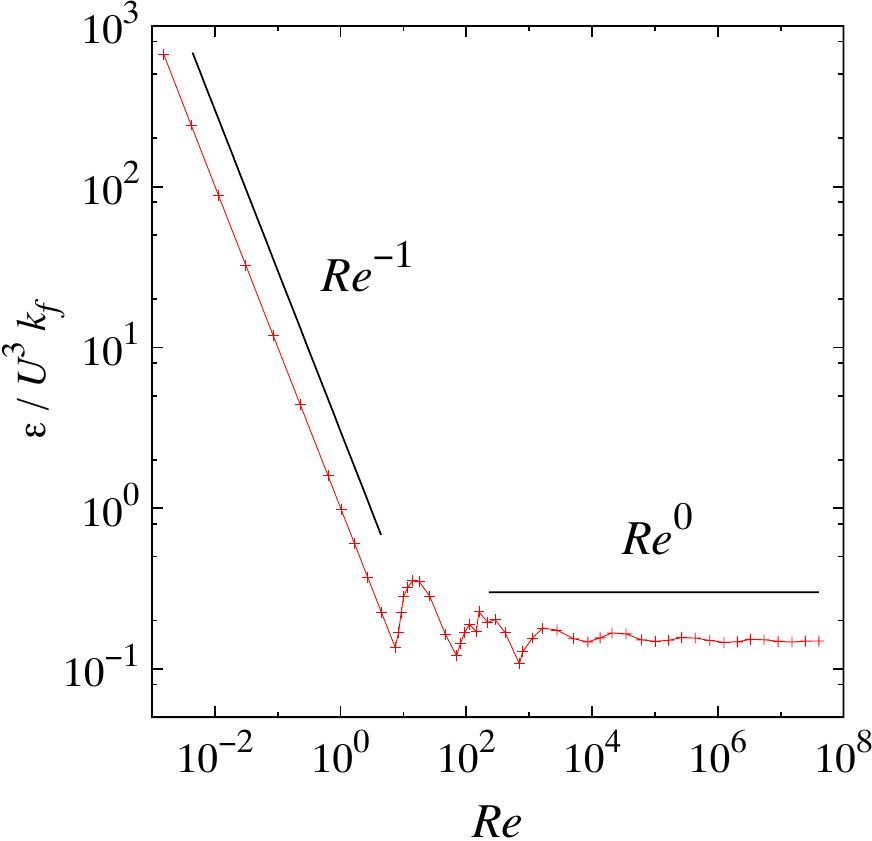}
\caption{Rescaled time-averaged energy dissipation rate as a function of the Reynolds number.}
\label{fig:1}
\end{figure}
\begin{figure}
\setlength{\unitlength}{\textwidth}
\includegraphics[width=0.488\textwidth]{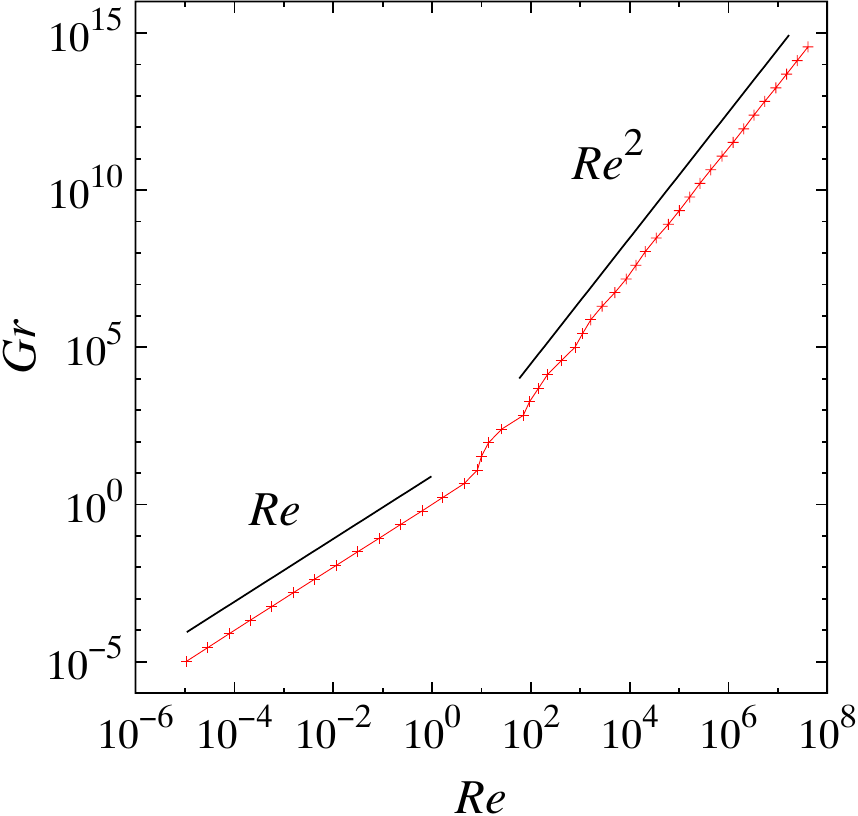}\hfill
\includegraphics[width=0.512\textwidth]{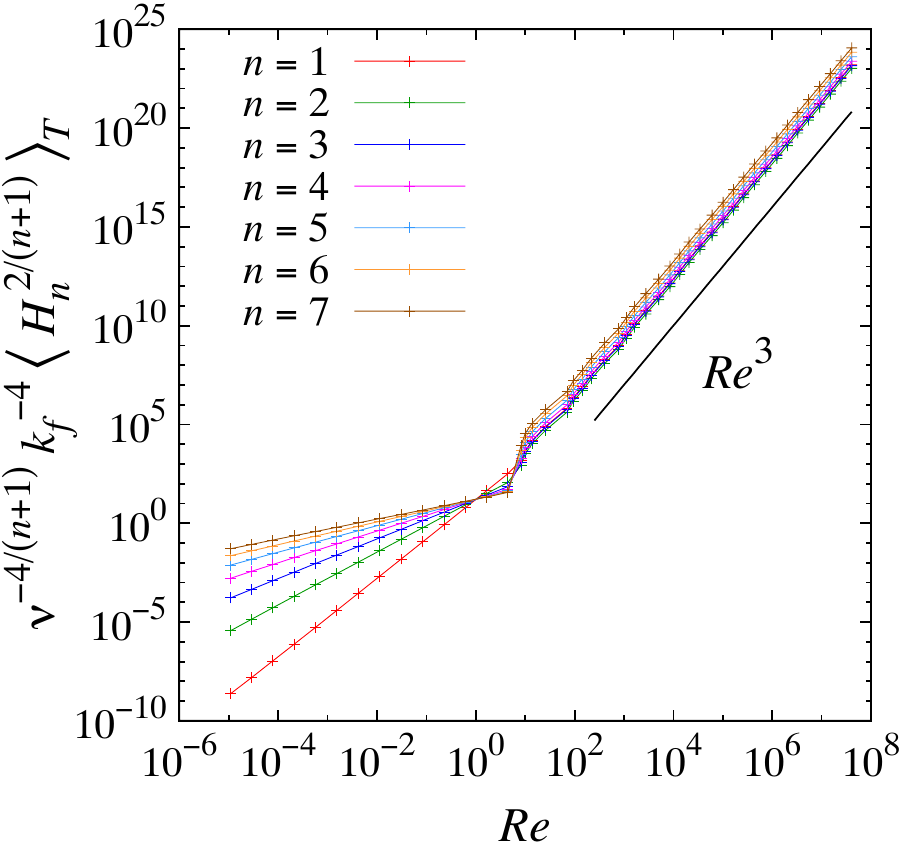}
\caption{(a) Grashof number as a function of the Reynolds number; (b) Time average of $H_n^\frac{2}{n+1}$ rescaled by
$\nu^\frac{4}{n+1} k_f^4$ as a function of the Reynolds number.}
\label{fig:2}
\end{figure}
\begin{figure}
\setlength{\unitlength}{\textwidth}
\includegraphics[width=0.502\textwidth]{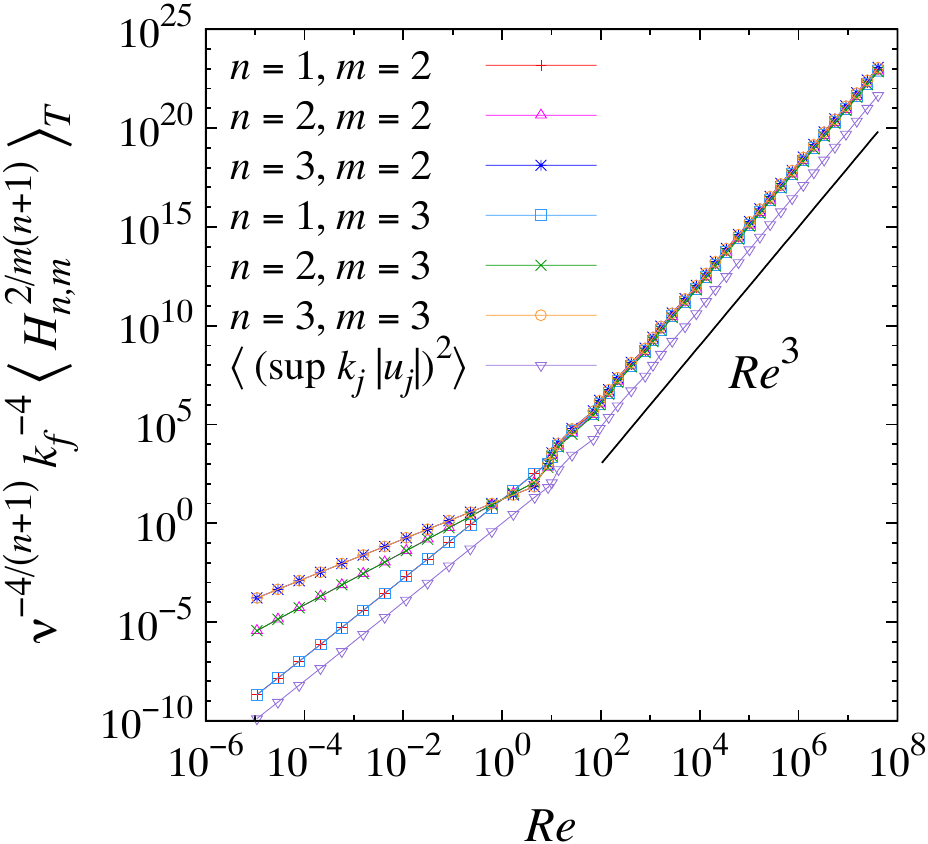}\hfill%
\includegraphics[width=0.498\textwidth]{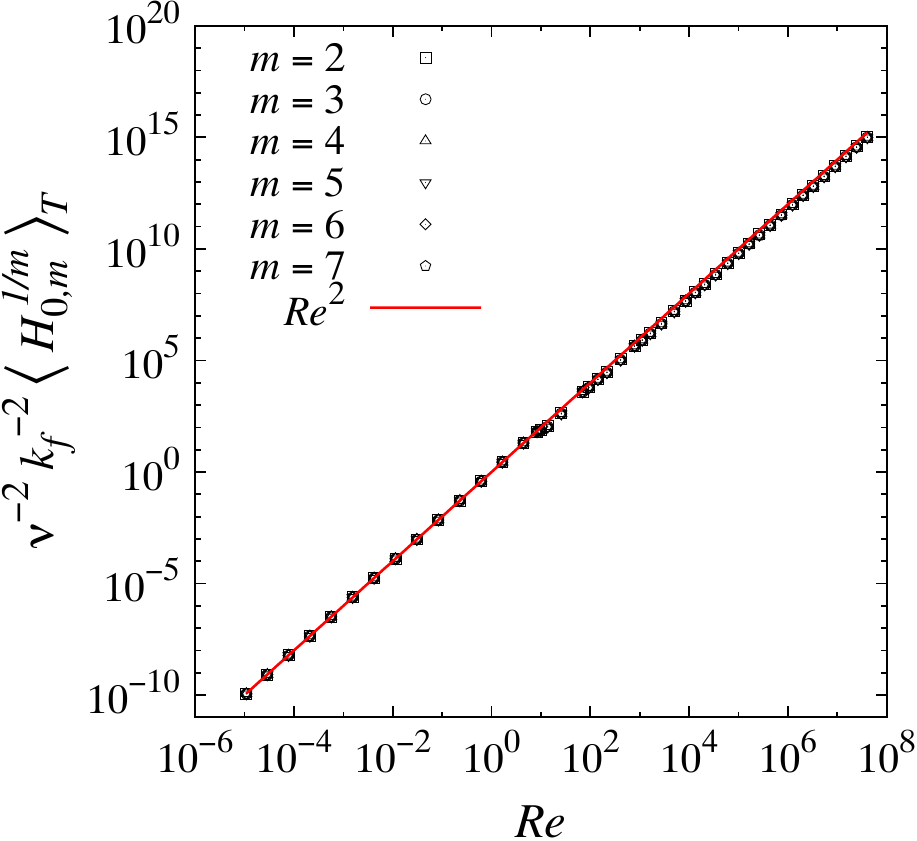}%
\caption{Time averages of (a) $H_{n,m}^\frac{2}{m(n+1)}$ rescaled by
$\nu^\frac{4}{n+1} k_f^4$ and (b) $H_{0,m}^\frac{1}{m}$ rescaled by $\nu^2 k_f^2$ as a function of the Reynolds number.}
\label{fig:3}
\end{figure}


\ack
The authors are grateful to Samriddhi Sankar Ray for several useful suggestions.
John Gibbon acknowledges the award of a Visiting Professorship at the Universit\'e C\^ote d’Azur during the months of November 2018 and April 2019 and the kind hospitality of the Laboratoire J. A. Dieudonn\'e.

\appendix
\section*{Appendix}
\setcounter{section}{1}


Inequality \eref{eq:triangle-ineq} is proved by induction on $n$ \cite{dg95}. By using the Cauchy--Schwarz inequality, we find
\bel
H_1 = \Sum{j} \vert u_j^\phexp\vert  \left( k_j^2 \vert u_j^\phexp\vert \right) \leqslant H_0^\frac{1}{2} H_{2}^\frac{1}{2} \,.
\ee
We then assume
\bel
H_n \leqslant H_0^\frac{1}{n+1} H_{n+1}^\frac{n}{n+1} 
\ee
and estimate $H_{n+1}$ as
\bel
H_{n+1} \leqslant H_{n \vphantom{+2}}^\frac{1}{2} H_{n+2}^\frac{1}{2} \leqslant H_0^\frac{1}{2(n+1)} H_{n+1}^\frac{n}{2(n+1)}  H_{n+2}^\frac{1}{2}\,,
\ee
which yields
\bel
H_{n+1} \leqslant H_0^\frac{1}{n+2} H_{n+2}^\frac{n+1}{n+2} \,.
\ee
This completes the proof by induction.

\end{document}